\documentclass[10pt,notitlepage,11pt,reqno]{amsart}

\usepackage{amssymb,amsmath,amsthm,mathrsfs,xspace,multicol}
\usepackage[mathscr]{eucal}
\usepackage{amscd}
\usepackage[breaklinks=true]{hyperref}

\newcommand{\comment}[1]{}
\newcommand{\blankbrac}{[ \cdot,\cdot ]}
\newcommand{\brac}[2]{\left[ #1,#2 \right]}
\newcommand{\Hblank}{\lbrace \cdot,\cdot \rbrace_{\mathrm{h}}}
\newcommand{\Sblank}{\lbrace \cdot,\cdot \rbrace_{\mathrm{s}}}
\newcommand{\Hbrac}[2]{\left \lbrace #1,#2 \right\rbrace_{\mathrm{h}}}
\newcommand{\Sbrac}[2]{\left \lbrace #1,#2 \right\rbrace_{\mathrm{s}}}
\newcommand{\h}{\mathrm{h}}
\newcommand{\s}{\mathrm{s}}
\newcommand{\ham}{\mathrm{Ham}\left(X\right)}
\newcommand{\hamG}{\mathrm{Ham}\left(X\right)^G}
\newcommand{\Gham}{\mathrm{Ham}\left(G\right)}
\newcommand{\hamL}{\mathrm{Ham}\left(G\right)^{L}}
\newcommand{\cinf}{C^{\infty}\left(X\right)}
\newcommand{\cinfG}{C^{\infty}\left(X\right)^G}
\newcommand{\cinfL}{C^{\infty}\left(G\right)^{L}}
\newcommand{\lie}[2]{\mathcal{L}\,_{v_{#1}}\, #2}
\newcommand{\Lie}{\mathcal{L}}

\newcommand{\R}{\mathbb{R}}
\newcommand{\ip}[1]{\iota_{v_{#1}}}

\newcommand{\maps}{\colon}    

\newcommand{\pback}[1]{{\mu_{#1}}^{\ast}}
\newcommand{\pfor}[1]{{\mu_{#1}}_{\ast}}
\newcommand{\g}{\mathfrak{g}}
\newcommand{\innerprod}[2]{\langle #1,#2 \rangle}

\newcommand{\Vect}{\mathrm{Vect}_{H}}

\newcommand{\tensor}{\otimes}

\DeclareMathOperator{\Aut}{\mathrm{Aut}}
\DeclareMathOperator{\End}{\mathrm{End}}

\DeclareMathOperator{\ad}{\mathrm{ad}}
\DeclareMathOperator{\Ad}{\mathrm{Ad}}

\theoremstyle{plain}
\newtheorem{theorem}{Theorem}[section]
\newtheorem{prop}[theorem]{Proposition}

\newtheorem{corollary}[theorem]{Corollary}
\newtheorem{definition}[theorem]{Definition}

\theoremstyle{remark}

\title{Categorified Symplectic Geometry and the String Lie 2-Algebra}
\author{John C.\ Baez and Christopher L.\ Rogers}
\email{\texttt{baez@math.ucr.edu},\texttt{chris@math.ucr.edu}}
\address{Department of Mathematics, University of California, 
Riverside, California 92521, USA}
\date{\today}
\begin{document}

\begin{abstract}
Multisymplectic geometry is a generalization of symplectic geometry
suitable for $n$-dimensional field theories, in which the
nondegenerate 2-form of symplectic geometry is replaced by a
nondegenerate $(n+1)$-form.  The case $n = 2$ is relevant to string
theory: we call this `2-plectic geometry'.  Just as the Poisson
bracket makes the smooth functions on a symplectic manifold into a Lie
algebra, the observables associated to a 2-plectic manifold form a
`Lie 2-algebra', which is a categorified version of a Lie algebra.
Any compact simple Lie group $G$ has a canonical 2-plectic structure,
so it is natural to wonder what Lie 2-algebra this example yields.
This Lie 2-algebra is infinite-dimensional, but we show here that the
sub-Lie-2-algebra of left-invariant observables is finite-dimensional,
and isomorphic to the already known `string Lie 2-algebra' associated
to $G$.  So, categorified symplectic geometry gives a geometric
construction of the string Lie 2-algebra.
\end{abstract}

\maketitle

\section{Introduction}
\label{introduction}

Symplectic geometry is part of a more general subject called
multisymplectic geometry, invented by DeDonder \cite{DeDonder} and
Weyl \cite{Weyl} in the 1930s.  In particular, just as the phase space
of a classical point particle is a symplectic manifold, a classical
string may be described using a finite-dimensional `2-plectic'
manifold.  Here the nondegenerate closed 2-form familiar from
symplectic geometry is replaced by a nondegenerate closed $3$-form.

Just as the smooth functions on a symplectic manifold form a Lie
algebra under the Poisson bracket operation, any 2-plectic manifold
gives rise to a `Lie 2-algebra'.  This is a {\it categorified} version
of a Lie algebra: that is, a category equipped with a bracket
operation obeying the usual Lie algebra laws {\it up to isomorphism}.
Alternatively, we may think of a Lie 2-algebra as a 2-term chain
complex equipped with a bracket satisfying the Lie algebra laws up to
chain homotopy.  

Now, every compact simple Lie group $G$ has a canonical 2-plectic
structure, built from the Killing form and the Lie bracket.  
Which Lie 2-algebra does this example yield?  Danny
Stevenson suggested that the answer should be related to the already 
known `string Lie 2-algebra' of $G$.  Our main result here confirms his
intuition.  The Lie 2-algebra associated to the 2-plectic manifold $G$
comes equipped with an action of $G$ via left translations.  The
translation-invariant elements form a Lie 2-algebra in their own
right, and this is the string Lie 2-algebra.  

This gives a new geometric construction of the string Lie 2-algebra.
For another construction, based on gerbes --- or alternatively,
central extensions of loop groups --- see the paper by Baez, Crans,
Schreiber and Stevenson \cite{BCSS}.  It will be interesting to see
what can be learned from comparing these approaches.

The plan of the paper is as follows.  In Section \ref{2-plectic} we
begin with a review of our recent paper \cite{BHR} on 2-plectic
geometry, Lie 2-algebras, and the classical bosonic string.  The goal
is to describe the Lie 2-algebra associated to a 2-plectic manifold.
We only state theorems needed for the work at hand, referring the
reader to the previous work for proofs and background material.  More
information on multisymplectic geometry can be found in papers by
Cantrijn, Ibort, and De Leon \cite{Cantrijn:1999} and by Cari\~{n}ena,
Crampin, and Ibort \cite{Carinena-Crampin-Ibort}.  Further details
regarding the application of multisymplectic geometry to classical
field theory can be found in the work of Kijowski \cite{Kijowski},
Gotay, Isenberg, Marsden, and Montgomery \cite{GIMM}, Helein
\cite{Helein}, and Rovelli \cite{Rovelli}.  For Lie 2-algebras, see
Baez and Crans \cite{HDA6} and also Roytenberg \cite{Roytenberg},
whose approach we will follow in this paper.

In Section \ref{group_actions} we consider a Lie group acting on a
2-plectic manifold, preserving the 2-plectic structure.  This group
then acts on the associated Lie 2-algebra, and the invariant elements
form a Lie sub-2-algebra.  In Section \ref{2-plectic_CSLG} we apply
this idea to the canonical 2-plectic structure on a compact Lie group
$G$, where $G$ acts on itself by left translations.  Finally, in
Section \ref{string_Lie} we show that the resulting Lie 2-algebra is
none other than the string Lie 2-algebra.

\section{2-Plectic Geometry and Lie 2-Algebras}
\label{2-plectic}

We begin by defining 2-plectic manifolds and Lie 2-algebras.
Then we explain how a 2-plectic manifold gives a Lie 2-algebra.

\begin{definition}
\label{2-plectic_def}
A 3-form $\omega$ on a $C^\infty$ manifold $X$ is 
{\bf 2-plectic}, or more specifically
a {\bf 2-plectic structure}, if it is both closed:
\[
    d\omega=0,
\]
and nondegenerate:
\[
    \forall v \in T_{x}X,\ \iota_{v} \omega =0 \Rightarrow v =0
\]
where we use $\iota_v \omega$ to stand for the interior product
$\omega(v, \cdot, \cdot)$.  
If $\omega$ is a 2-plectic form on $X$ we call the pair $(X,\omega)$ 
a \bf{2-plectic manifold}.
\end{definition}

Note that the 2-plectic structure induces an injective map from the
space of vector fields on $X$ to the space of 2-forms on $X$. This leads
us to the following definition:

\begin{definition}
Let $(X,\omega)$ be a 2-plectic manifold.  A 1-form $\alpha$ on $X$ 
is {\bf Hamiltonian} if there exists a vector field $v_\alpha$ on $X$ such that
\[
d\alpha= -\ip{\alpha} \omega.
\]
We say $v_\alpha$ is the {\bf Hamiltonian vector field} corresponding to $\alpha$. 
The set of Hamiltonian 1-forms and the set of Hamiltonian vector
fields on a 2-plectic manifold are both vector spaces and are denoted
as $\ham$ and $\Vect\left(X \right)$, respectively.
\end{definition}

The Hamiltonian vector field $v_\alpha$ is unique if it exists, but
note there may be 1-forms $\alpha$ having no Hamiltonian vector field.
Furthermore, two distinct Hamiltonian 1-forms may differ by a closed
1-form and therefore share the same Hamiltonian vector field.

We can generalize the Poisson bracket of functions in symplectic
geometry by defining a bracket of Hamiltonian 1-forms. This can be
done in two ways:

\begin{definition}  
\label{hemi-bracket.defn}
Given $\alpha,\beta \in \ham$, the {\bf hemi-bracket} 
$\Hbrac{\alpha}{\beta}$ is the 1-form given by
\[  \Hbrac{\alpha}{\beta} = \lie{\alpha}{\beta},\]
where $\Lie_{v_{\alpha}}$ is the Lie derivative along the vector field $v_{\alpha}$.
\end{definition}

\begin{definition}
\label{semi-bracket.defn}
Given $\alpha,\beta\in \ham$, the {\bf semi-bracket} $\Sbrac{\alpha}{\beta}$
is the 
\break
1-form given by 
\[  \Sbrac{\alpha}{\beta} = \ip{\beta}\ip{\alpha}\omega .\]
\end{definition}

These brackets in general will differ by an exact 1-form:

\begin{prop} Given $\alpha,\beta\in \ham$, 
\[
\Hbrac{\alpha}{\beta}=\Sbrac{\alpha}{\beta} + d\ip{\alpha}\beta. 
\label{bracket_relation}
\]
\end{prop}
\begin{proof}
See Proposition 3.5 in \cite{BHR}.
\end{proof}

The space $\ham$ is closed under both brackets, but neither bracket
satisfies all the axioms of a Lie algebra.  The hemi-bracket fails to
be skew-symmetric, while the semi-bracket fails to satisfy the Jacobi
identity.

\begin{prop} \label{hemi-bracket} Let $\alpha,\beta,\gamma \in \ham$ and
let $v_\alpha,v_\beta,v_\gamma$ be the respective Hamiltonian
vector fields.  The hemi-bracket $\Hblank$ has the following 
properties:

  \begin{enumerate}
\item The bracket of Hamiltonian forms is Hamiltonian:
  \begin{eqnarray}
    d\Hbrac{\alpha}{\beta} = -\iota_{[v_\alpha,v_\beta]} \omega
\label{hemi-closure}
  \end{eqnarray}
so in particular we have 
\[     v_{\Hbrac{\alpha}{\beta}} = [v_\alpha,v_\beta]  .\]

\item The bracket is skew-symmetric up to an exact 1-form:
  \begin{eqnarray}
    \Hbrac{\alpha}{\beta} + dS_{\alpha,\beta} = -\Hbrac{\beta}{\alpha} 
  \end{eqnarray}
  with $S_{\alpha,\beta}=-(\ip{\alpha}\beta + \ip{\beta}\alpha)$.

\item The bracket satisfies the Jacobi identity:
  \begin{eqnarray}
    \Hbrac{\alpha}{\Hbrac{\beta}{\gamma}} = 
\Hbrac{\Hbrac{\alpha}{\beta}}{\gamma} + \Hbrac{\beta}{\Hbrac{\alpha}{\gamma}}. 
  \end{eqnarray}
  
  \end{enumerate}
\end{prop}
\begin{proof}
See Proposition 3.6 in \cite{BHR}.
\end{proof}

\begin{prop}\label{semi-bracket} Let $\alpha,\beta,\gamma \in \ham$ and let
$v_\alpha,v_\beta,v_\gamma$ be the respective Hamiltonian
vector fields.  The semi-bracket $\Sblank$ has the following properties:
  \begin{enumerate}
\item The bracket of Hamiltonian forms is Hamiltonian:
  \begin{eqnarray}
    d\Sbrac{\alpha}{\beta} = -\iota_{[v_\alpha,v_\beta]} \omega.
  \end{eqnarray}
so in particular we have 
\[     v_{\Sbrac{\alpha}{\beta}} = [v_\alpha,v_\beta]  .\]
\item The bracket is antisymmetric: 
  \begin{eqnarray}
    \Sbrac{\alpha}{\beta} = -\Sbrac{\beta}{\alpha}
  \end{eqnarray}

\item The bracket satisfies the Jacobi identity up to an exact 1-form:
\begin{eqnarray}
    \Sbrac{\alpha}{\Sbrac{\beta}{\gamma}} + dJ_{\alpha,\beta,\gamma} =
    \Sbrac{\Sbrac{\alpha}{\beta}}{\gamma} +
    \Sbrac{\beta}{\Sbrac{\alpha}{\gamma}} 
  \end{eqnarray}
with $J_{\alpha,\beta,\gamma}=-\ip{\alpha}\ip{\beta}\ip{\gamma}\omega$.
\end{enumerate}
\end{prop}
\begin{proof}
See Proposition 3.7 in \cite{BHR}.
\end{proof}

The observation that these brackets satisfy the Lie algebra laws `up
to exact 1-forms' leads to the notion of a Lie 2-algebra. Here we
define a Lie 2-algebra to be a 2-term chain complex of vector spaces
equipped with structures analogous those of a Lie algebra, for which
the usual laws hold up to coherent chain homotopy. Alternative definitions
equivalent to the one given here are presented in \cite{HDA6} and
\cite{Roytenberg}.
\begin{definition}
A {\bf Lie 2-algebra} is a 2-term chain complex of vector spaces
$L = (L_0\stackrel{d}\leftarrow L_1)$ equipped with the following structure:
\begin{itemize}
\item a chain map $\blankbrac\maps L \tensor L\to L$ called the {\bf
bracket};
\item a chain homotopy $S \maps L\tensor L \to L$
from the chain map
\[     \begin{array}{ccl}  
     L \tensor L &\to& L   \\
     x \tensor y &\longmapsto& [x,y]  
  \end{array}
\]
to the chain map
\[     \begin{array}{ccl}  
     L \tensor L & \to & L   \\
     x \tensor y & \longmapsto & -[y,x]  
  \end{array}
\]
called the {\bf alternator};
\item an antisymmetric chain homotopy $J \maps L \tensor L \tensor L
  \to L$ 
from the chain map
\[     \begin{array}{ccl}  
     L \tensor L \tensor L & \to & L   \\
     x \tensor y \tensor z & \longmapsto & [x,[y,z]]  
  \end{array}
\]
to the chain map
\[     \begin{array}{ccl}  
     L \tensor L \tensor L& \to & L   \\
     x \tensor y \tensor z & \longmapsto & [[x,y],z] + [y,[x,z]]  
  \end{array}
\]
called the {\bf Jacobiator}.
\end{itemize}
In addition, the following equations are required to hold:
\begin{equation}
\begin{array}{c}
  [x,J(y,z,w)] + J(x,[y,z],w) +
  J(x,z,[y,w]) + [J(x,y,z),w] \\ + [z,J(x,y,w)] 
  = J(x,y,[z,w]) + J([x,y],z,w) \\ + [y,J(x,z,w)] + J(y,[x,z],w) + J(y,z,[x,w]),
\end{array}
\end{equation}
\begin{equation}
    J(x,y,z)+J(y,x,z)=-[S(x,y),z],
\end{equation}
\begin{equation}
    J(x,y,z)+J(x,z,y)=[x,S(y,z)]-S([x,y],z)-S(y,[x,z]),
\end{equation}
\begin{equation}
    {S(x,[y,z])} = S([y,z],x).
\end{equation}
\end{definition}

\begin{definition}
A Lie 2-algebra for which the Jacobiator is the identity
chain homotopy is called {\bf hemistrict}.  One for which
the alternator is the identity chain homotopy is called {\bf semistrict}.
\end{definition}
\noindent

Given a 2-plectic manifold $(X,\omega)$, we can 
construct both a hemistrict and a semistrict Lie 2-algebra.
Both of these Lie 2-algebras have the same underlying 2-term complex, 
namely:
\[
L \quad = \quad 
\ham  
\stackrel{d}{\leftarrow} \cinf  
\stackrel{0}{\leftarrow} 0 
\stackrel{0}{\leftarrow} 0 
\stackrel{0}{\leftarrow} \cdots 
 \]
where $d$ is the usual exterior derivative of functions.  
This chain complex is well-defined, since 
any exact form is Hamiltonian, with $0$ as its Hamiltonian vector
field.

The hemistrict Lie 2-algebra is equipped with a chain map called the
{\bf hemi-bracket}:
\[      \Hbrac{\cdot}{\cdot} \maps L \otimes L \to L  .\]
In degree $0$, the hemi-bracket is given as in 
Definition \ref{hemi-bracket.defn}: 
\[  \Hbrac{\alpha}{\beta} = \lie{\alpha}\beta.   \]
In degree $1$, it is given by: 
\[  \Hbrac{\alpha}{f} = \lie{\alpha}{f}, \qquad \Hbrac{f}{\alpha} = 0.  \]
In degree $2$, we necessarily have
\[   \Hbrac{f}{g} = 0.    \]
Here $\alpha,\beta \in \ham$, while $f,g \in \cinf$.


Similarly, the semistrict Lie 2-algebra comes with a 
chain map called the {\bf semi-bracket}:
\[      \Sbrac{\cdot}{\cdot} \maps L \otimes L \to L  .\]
In degree $0$, the semi-bracket is given as in 
Definition \ref{semi-bracket.defn}: 
\[  \Sbrac{\alpha}{\beta} = \ip{\beta}\ip{\alpha}\omega .\]
In degrees $1$ and $2$, we set it equal to zero:
\[  \Hbrac{\alpha}{f} = 0, \qquad \Hbrac{f}{\alpha} = 0, \qquad
    \Hbrac{f}{g} = 0.    \]

The precise constructions of these Lie 2-algebras are 
given as follows:

\begin{theorem}
\label{hemistrict}
If $(X,\omega)$ is a 2-plectic manifold, there is 
a hemistrict Lie 2-algebra $L(X,\omega)_\h$
where:
\begin{itemize}
\item the space of 0-chains is $\ham$,
\item the space of 1-chains is $\cinf$,
\item the differential is the exterior derivative $d \maps \cinf \to \ham$,
\item the bracket is $\Hblank$,
\item the alternator is the bilinear map $S \maps \ham \times \ham \to \cinf$ 
defined by $S_{\alpha,\beta}= -(\ip{\alpha}\beta + \ip{\beta}\alpha)$, and
\item the Jacobiator is the identity chain homotopy, hence given by the 
trilinear map $J \maps \ham \times \ham \times \ham \to \cinf$ with 
$J_{\alpha,\beta,\gamma} = 0$.
\end{itemize}
\end{theorem}
\begin{proof}
See Theorem 4.3 in \cite{BHR}.
\end{proof}

\begin{theorem}
\label{semistrict}
If $(X,\omega)$ is a 2-plectic manifold, there is a 
semistrict Lie 2-algebra $L(X,\omega)_\s$ where:
\begin{itemize}
\item the space of 0-chains is $\ham$,
\item the space of 1-chains is $\cinf$,
\item the differential is the exterior derivative $d \maps \cinf \to \ham$,
\item the bracket is $\Sblank$,
\item the alternator is the identity chain map, hence given by the bilinear
map $S \maps \ham \times \ham \to \cinf$ with $S_{\alpha,\beta} = 0$, and
\item the Jacobiator is the trilinear map $J\maps \ham\times \ham\times 
\ham\to \cinf$ defined by $J_{\alpha,\beta,\gamma} = -\ip{\alpha}\ip{\beta}\ip{\gamma}\omega$.
\end{itemize}
\end{theorem}
\begin{proof}
See Theorem 4.4 in \cite{BHR}.
\end{proof}

A Lie 2-algebra homomorphism is a chain map between the underlying
chain complexes that preserves the bracket up to `coherent chain homotopy'.
More precisely:

\begin{definition}
\label{homo}
Given Lie 2-algebras $L$ and $L'$ with bracket, alternator and 
Jacobiator $\blankbrac$, $S$, $J$ and $\blankbrac^\prime$, $S^\prime$, 
$J^\prime$ respectively, a {\bf homomorphism} from $L$ to $L'$
consists of:
\begin{itemize}
\item{a chain map $\phi=\left(\phi_{0},\phi_{1}\right) \maps L \to L'$, and}
\item{a chain homotopy $\Phi \maps L \tensor L \to L$ from the chain
  map
\[     \begin{array}{ccl}  
     L \tensor L & \to & L^{\prime}   \\
     x \tensor y & \longmapsto & \left [ \phi(x),\phi(y) \right]^{\prime},
  \end{array}
\]
to the chain map
\[     \begin{array}{ccl}  
     L \tensor L & \to & L^{\prime}   \\
     x \tensor y & \longmapsto & \phi \left( [x,y] \right)
  \end{array}
\]
}
\end{itemize}
such that the following equations hold:
\begin{equation}
    S^{\prime}\left(\phi_{0}(x),\phi_{0}(y)\right)-\phi_{1}(S\left(x,y
    \right))=\Phi(x,y)+\Phi(y,x),
\end{equation}
\begin{equation}
     \begin{array}{l}
       J^{\prime}\left(\phi_{0}(x),\phi_{0}(y),\phi_{0}(z)\right)-\phi_{1}\left(J\left(
       x,y,z \right)\right)\\
       =[\phi_{0}(x),\Phi(y,z)]^{\prime}-[\phi_{0}(y),\Phi(x,z)]^{\prime}-[\Phi(x,y),\phi_{0}(z)]^{\prime}-\\
       -\Phi([x,y],z)-\Phi(y,[x,z])+\Phi(x,[y,z]).
     \end{array}
\end{equation}
\end{definition}

The details involved in composing Lie 2-algebra homomorphisms are given
by Roytenberg \cite{Roytenberg}.  We say a Lie 2-algebra homomorphism
with an inverse is an {\bf isomorphism}.

In fact, our previous work \cite{BHR} shows that the hemistrict and
semistrict Lie 2-algebras associated to a 2-plectic manifold are
isomorphic:

\begin{theorem}
\label{isomorphism}
There is a Lie 2-algebra isomorphism 
\[       \phi \maps L(X,\omega)_\h \to L(X,\omega)_\s  \]
given by the identity chain map and a chain 
homotopy $\Phi \maps L \tensor L \to L$ that is nontrivial only 
in degree 0, where it is given by
\[\Phi \left(\alpha,\beta \right)=\ip{\alpha}\beta\] 
for $\alpha,\beta \in L_{0}=\ham$. 
\end{theorem}
\begin{proof}
See the proof of Theorem 4.6 in \cite{BHR}.
\end{proof}

So, while the semistrict and hemistrict Lie 2-algebras defined above
look different at first sight, we may legitimately speak of `the' Lie
2-algebra associated to a 2-plectic manifold.  We shall refer back to
this result several times in the subsequent sections.

\section{Group Actions on 2-Plectic Manifolds}
\label{group_actions}

Next suppose we have a Lie group acting on a 2-plectic manifold,
preserving the 2-plectic structure.  In this situation both the
hemistrict and semistrict Lie 2-algebras constructed above have Lie
sub-2-algebras consisting of {\it invariant} Hamiltonian forms and
functions.

More precisely, let
$\mu \maps G \times X \to X$ be a left action of the Lie group $G$ on the
2-plectic manifold $\left(X,\omega \right)$, and assume this action
preserves the 2-plectic structure:
\[\mu_{g}^{\ast} \omega = \omega\]
for all $g \in G$.  Denote the subspace of invariant Hamiltonian 1-forms 
as follows:
\[\hamG = \left \{ \alpha \in \ham ~ \vert ~ \forall g \in G ~
\mu_{g}^{\ast}\alpha = \alpha \right\}.\]
The Hamiltonian vector field of an
invariant Hamiltonian 1-form is itself invariant under the
action of $G$:
\begin{prop}
\label{invariant_vector_fields}
If $\alpha \in \hamG$ and $v_{\alpha}$ is the Hamiltonian vector field
associated with $\alpha$, then $\pfor{g} v_{\alpha} = v_{\alpha}$ for
all $g \in G$.
\end{prop}

\begin{proof}
The exterior derivative commutes with the pullback of the group
action. Therefore if $v_{1}, v_{2}$ are smooth vector fields, then
$d\alpha \left(\pfor{g} v_{1},\pfor{g} v_{2}\right)=
d\alpha\left(v_{1},v_{2} \right)$, since we are assuming $\alpha$
is $G$-invariant.  Since $\alpha \in \ham$, then 
$d \alpha = -\ip{\alpha}\omega$, so 
\[
\omega\left(v_{\alpha},\pfor{g} v_{1},\pfor{g} v_{2} \right)
= \omega \left(v_{\alpha},v_{1},v_{2} \right)=
\omega\left(\pfor{g}v_{\alpha},\pfor{g} v_{1},\pfor{g} v_{2} \right),
\]
where the last equality follows from $\pback{g}\omega=\omega$.
Therefore 
\[
\omega\left(v_{\alpha}-\pfor{g}v_{\alpha},\pfor{g} v_{1},\pfor{g} v_{2} \right)=0.
\]
Since $\omega$ is nondegenerate, and $v_{1}, v_{2}$ are arbitrary, it
follows that $\pfor{g}v_{\alpha} = v_{\alpha}$.
\end{proof}

Let $\cinfG$ denote the subspace of invariant smooth functions on $X$:
\[\cinfG = \left\{ f \in \cinf ~ \vert ~ \forall g \in G ~
\mu_{g}^{\ast}f = f \circ \mu_{g} = f \right\},\] 
and let $L^{G}$ denote the 2-term complex composed of $\hamG$ and $\cinfG$:
\[
L^{G} =  
\hamG  
\stackrel{d}{\leftarrow} \cinfG,  
 \]
where $d$ is the exterior derivative. 

The invariant differential forms on $X$ form a graded subalgebra that is
stable under exterior derivative and interior product by an invariant
vector field. Since the hemi-bracket and semi-bracket
introduced in Definitions \ref{hemi-bracket.defn} and
\ref{semi-bracket.defn} are nothing but compositions of these
operations, they restrict to well-defined chain maps:
\[\Hblank \maps L^{G} \otimes L^{G} \to L^{G},\] 
\[\Sblank \maps L^{G} \otimes L^{G} \to L^{G}.\]   
More precisely, we have the following proposition:
\begin{prop}
\label{bracket_closure}
If $\alpha,\beta \in \hamG$ and $f \in \cinfG$, then:
\begin{itemize}
\item{$\Hbrac{\alpha}{f},\Hbrac{f}{\alpha} \in \cinfG$,}
\item{$\Sbrac{\alpha}{f},\Sbrac{f}{\alpha} \in \cinfG$,}
\item{$\Hbrac{\alpha}{\beta}\in \hamG$,}
\item{$\Sbrac{\alpha}{\beta}\in \hamG$.} 
\end{itemize}
\end{prop}

\begin{proof}
By definition,
$\Hbrac{f}{\alpha}=\Sbrac{\alpha}{f}=\Sbrac{f}{\alpha}=0$. Therefore
they are trivially invariant under the group action. If $\alpha$ is an
invariant Hamiltonian 1-form, then by Proposition
\ref{invariant_vector_fields} its Hamiltonian vector field
$v_{\alpha}$ is invariant. Since $\Hbrac{\alpha}{\beta}=
\lie{\alpha}{\beta}= d\ip{\alpha}\beta+\ip{\alpha}d\beta$ and
$\Sbrac{\alpha}{\beta}=\ip{\beta}\ip{\alpha}\omega$, it then follows
from the above remarks that $\Hbrac{\alpha}{\beta}$ and
$\Sbrac{\alpha}{\beta}$ are themselves invariant.
\end{proof}

An important consequence of Proposition \ref{bracket_closure} is that
$L^G$ forms a Lie sub-2-algebra of the hemistrict Lie 2-algebra
$L(X,\omega)_\h$, which we call $L(X,\omega)_\h^G$.  It also forms a
Lie sub-2-algebra of $L(X,\omega)_\s$, which we call
$L(X,\omega)_\s^G$.  Furthermore, the isomorphism
\[   L(X,\omega)_\h \cong L(X,\omega)_\s  \]
of Theorem \ref{isomorphism} restricts to an isomorphism between these
Lie sub-2-algebras:
\[  L(X,\omega)^{G}_\h \cong L(X,\omega)^{G}_\s .\]
We summarize these results in the next three corollaries:

\begin{corollary}
\label{hemistrictG}
If $\mu \maps G \times X \to X$ is a left action of the Lie group $G$ on the
2-plectic manifold $\left(X,\omega \right)$ and for all $g \in
G, ~ \mu_{g}^{\ast} \omega = \omega$, then there is 
a hemistrict Lie 2-algebra $L(X,\omega)^{G}_\h$,
where:
\begin{itemize}
\item the space of 0-chains is $\hamG$,
\item the space of 1-chains is $\cinfG$,
\item the differential is the exterior derivative $d \maps \cinfG \to \hamG$,
\item the bracket is $\Hblank$,
\item the alternator is the bilinear map $S \maps \hamG \times \hamG \to \cinfG$ 
defined by $S_{\alpha,\beta}= -(\ip{\alpha}\beta + \ip{\beta}\alpha)$, and
\item the Jacobiator is the identity, hence given by the trilinear map
$J \maps \hamG \times \hamG \times \hamG \to \cinfG$ with 
$J_{\alpha,\beta,\gamma} = 0$.
\end{itemize}
\end{corollary}

\begin{corollary}
\label{semistrictG}
If $\mu \maps G \times X \to X$ is a left action of the Lie group $G$ on the
2-plectic manifold $\left(X,\omega \right)$ and for all $g \in
G, \mu_{g}^{\ast} \omega = \omega$, then there is a 
semistrict Lie 2-algebra $L(X,\omega)^{G}_\s$ where:
\begin{itemize}
\item the space of 0-chains is $\hamG$,
\item the space of 1-chains is $\cinfG$,
\item the differential is the exterior derivative $d \maps \cinfG \to \hamG$,
\item the bracket is $\Sblank$,
\item the alternator is the identity, hence given by the bilinear
map $S \maps \hamG \times \hamG \to \cinfG$ with $S_{\alpha,\beta} = 0$, and
\item the Jacobiator is the trilinear map $J\maps \hamG\times \hamG\times 
\hamG\to \cinfG$ defined by $J_{\alpha,\beta,\gamma} = -\ip{\alpha}\ip{\beta}\ip{\gamma}\omega$.
\end{itemize}
\end{corollary}

\begin{corollary}
\label{isomorphismG}
The isomorphism 
\[       \phi \maps L(X,\omega)_\h \to L(X,\omega)_\s  \]
restricts to a Lie 2-algebra isomorphism
\[       \phi \maps L(X,\omega)^{G}_\h \to L(X,\omega)^{G}_\s .\]
\end{corollary}
\begin{proof}
As mentioned in Theorem \ref{isomorphism},
the isomorphism between $L(X,\omega)_\h$ and $L(X,\omega)_\s$ is given
by the identity chain map and a chain homotopy $\Phi$
that is nontrivial only in degree 0: 
\[\Phi \left(\alpha,\beta \right)=\ip{\alpha}\beta,\] 
where $\alpha,\beta \in L_{0}=\ham$. If $\alpha$ and $\beta$ are
invariant under the action of $G$, then it follows from Proposition
\ref{invariant_vector_fields} that $\ip{\alpha}\beta$ is an invariant
smooth function. Hence the chain homotopy $\Phi$ restricts to a chain
homotopy $\Phi \maps L^{G} \tensor L^{G} \to L^{G}$.
\end{proof}

\section{The 2-Plectic Structure on a Compact Simple Lie Group}
\label{2-plectic_CSLG}

Every compact simple Lie group has a canonical 2-plectic
structure. This structure has been discussed in the multisymplectic
geometry literature \cite{Cantrijn:1999,Ibort:2000}, and plays a
crucial role in several branches of mathematics connected to string
theory, including the theory of affine Lie algebras, central
extensions of loop groups, gerbes, and Lie 2-groups
\cite{BCSS,Brylinski,CJMSW,PressleySegal}.

Recall that if $G$ is a compact Lie group, then its Lie algebra $\g$ 
admits an inner product $\innerprod{\cdot}{\cdot}$ that is invariant
under the adjoint representation $\Ad \maps G \to \Aut\left(\g
\right)$.  
For any nonzero real number $k$, we can define a trilinear form
\[ \theta_{k}(x,y,z)= k \innerprod{x}{[y,z]}\]
for any $x,y,z \in \g$. 
Since the inner product is invariant under the adjoint representation,
it follows that the linear transformations $\ad_y \maps \g \to \g$
given by $\ad_y(x) = [y,x]$ are skew adjoint. 
That is,
$\innerprod{\ad_{y}(x)}{z}= -\innerprod{x}{\ad_{y}(z)}$ for all
$x,y,z \in \g$. It follows that $\theta_{k}$ is totally antisymmetric.
Moreover, $\theta_k$ is invariant under the adjoint representation since
$\left[\Ad_{g}(x),\Ad_{g}(y) \right] = \Ad_{g}\left([x,y] \right)$.

Let $L_{g} \maps G \to G$ and $R_{g} \maps G \to G$ denote left and
right translation by $g$, respectively. Using left translation, we can
extend $\theta_{k}$ to a left invariant 3-form $\nu_{k}$ on $G$. More
precisely, given $g \in G$ and $v_{1},v_{2},v_{3} \in T_{g} G$ define
a smooth section $\nu_{k}$ of $\Lambda^3 T^{\ast}G$ by
\[ \nu_{k} \vert_{g} \left(v_{1},v_{2},v_{3} \right) = 
\theta_{k} \left(L_{g^{-1}\ast} v_{1},L_{g^{-1}\ast} v_{2},L_{g^{-1}\ast}
v_{3} \right).\] 
It is straightforward to show that $\nu_{k}$ is also a right invariant
3-form. Indeed, since $\Ad_{g} =L_{g \ast} \circ R_{g^{-1} \ast}$, the
invariance of $\theta_{k}$ under the adjoint representation implies 
$R^{\ast}_{g}\nu_{k} =\nu_{k}$. From the left and right invariance we
can conclude
\[d\nu_{k}=0\] 
since any $p$-form on a Lie group that is both left and right 
invariant is closed. 

Now suppose that $G$ is a compact simple Lie group. Then $\g$ is
simple, so it has a canonical invariant inner product: the Killing 
form, normalized to taste.  With this choice of inner product, 
the trilinear form $\theta_{k}$ is nondegenerate in the sense of 
Definition \ref{2-plectic_def}:
\begin{prop} \label{nondegenerate}
Let $G$ be a compact simple Lie group with Lie algebra $\g$. If $x \in
\g$ and $\theta_{k}(x,y,z)=0$ for all $y,z \in \g$ then $x=0$.
\end{prop}
\begin{proof}
Recall that if $\g$ is simple, then it is equal to its derived algebra
$\brac{\g}{\g}$. Therefore there exists $y,z \in \g$ such that
$x=\brac{y}{z}$. Since $\innerprod{\cdot}{\cdot}$ is an inner product, 
$\theta_{k}(x,y,z)=k \innerprod{x}{[y,z]}=k \innerprod{x}{x}=0$
implies $x=0$.
\end{proof}
It is easy to see that the nondegeneracy of $\theta_{k}$ implies the
nondegeneracy of $\nu_{k}$. Therefore $\nu_{k}$ is a closed,
nondegenerate 3-form. We summarize the preceding discussion in the
following proposition:

\begin{prop} \label{G_is_2-plectic}
Let $G$ be a compact simple Lie group with Lie algebra $\g$. Let 
$\innerprod{\cdot}{\cdot}$ be the Killing form, and let $k$ be a
nonzero real number.  The left-invariant 3-form $\nu_{k}$ on $G$ 
corresponding to $\theta_k \in \Lambda^3 \g^*$ is 2-plectic.
So, $\left(G,\nu_{k} \right)$ is a 2-plectic manifold.
\end{prop}

Now we wish to identify the Hamiltonian 1-forms associated with the
2-plectic structure $\nu_{k}$ that are invariant under left
translation.  We denote the space of all left invariant 1-forms as
$\g^{\ast}$. The left invariant Hamiltonian 1-forms, their
corresponding Hamiltonian vector fields, and the left invariant smooth
real-valued functions will be denoted as $\hamL$, $\Vect
\left(G\right)^L$, and $\cinfL$, respectively.

If $f \in \cinfL$, then by definition $f=f \circ L_{g}$ for all $g \in
G$. Hence $f$ must be a constant function, so $\cinfL$ may be identified 
with $\R$.  The following theorem describes the left invariant Hamiltonian 
1-forms:
\begin{theorem}
\label{left_invariant_1-forms}
Every left invariant 1-form on $\left(G,\nu_{k}\right)$ is
Hamiltonian. That is, $\hamL = \g^{\ast}$.
\end{theorem} 
\begin{proof}
Recall that if $\alpha$ is a smooth 1-form and $v_{0},v_{1}$
are smooth vector fields on any manifold, then
\[d\alpha\left(v_{0},v_{1}\right)= v_{0} \left(\alpha
\left(v_{1}\right)\right) - v_{1}\left(\alpha\left(v_{0}\right)\right)
- \alpha \left(\left[v_{0},v_{1}\right]\right).\]
Suppose now that $\alpha$ is a left invariant 1-form on $G$ and 
$v_{0},v_{1}$ are left invariant vector fields.  Then the smooth functions
$\alpha\left(v_{1}\right)$ and $\alpha\left(v_{0}\right)$ are also left 
invariant and therefore constant. Therefore the right hand side of
the above equality simplifies and we have
\[d\alpha\left(v_{0},v_{1}\right)= - \alpha \left(\left[v_{0},v_{1}\right]\right).\]

Let $\alpha \in \g^{\ast}$ and let $\innerprod{\cdot}{\cdot}$ be the
inner product on $\g$ used in the construction of
$\nu_{k}$. There exists a left invariant vector field $v_{\alpha} \in \g$
such that $\alpha(x)=k\innerprod{v_{\alpha}}{x}$ for all $x \in
\g$. Combining this with the above expression for $d \alpha$ gives
\[d \alpha\left(x,y\right)=-k\innerprod{v_{\alpha}}{[x,y]},\]
which implies
\[d \alpha = -\ip{\alpha}\nu_{k}.\] 
Hence $\alpha \in \Gham$, and $\hamL=\Gham \cap \g^{\ast} =\g^{\ast}$.
\end{proof}

The most important application of Theorem \ref{left_invariant_1-forms}
is that it allows us to use Corollaries \ref{hemistrictG} and
\ref{semistrictG} to construct hemistrict and semistrict Lie
2-algebras both having $\g^{\ast}$ as their space of 0-chains, 
where $\g$ is the Lie algebra of a compact simple Lie group.  
We summarize these facts in the following two corollaries:

\begin{corollary}
\label{hemistrictG2}
If $G$ is a compact simple Lie group with Lie algebra $\g$ and 2-plectic structure $\nu_{k}$,
then there is a hemistrict Lie 2-algebra $L(G,k)_\h$
where:
\begin{itemize}
\item the space of 0-chains is $\g^{\ast}$,
\item the space of 1-chains is $\R$,
\item the differential is the exterior derivative $d \maps \R \to
  \g^{\ast}$ (i.e.\ $d=0$), 
\item the bracket is $\Hblank$,
\item the alternator is the bilinear map $S \maps \g^{\ast} \times \g^{\ast} \to \R$ 
defined by $S_{\alpha,\beta}= -(\ip{\alpha}\beta + \ip{\beta}\alpha)$, and
\item the Jacobiator is the identity, hence given by the trilinear map
$J \maps \g^{\ast} \times \g^{\ast} \times \g^{\ast} \to \R$ with 
$J_{\alpha,\beta,\gamma} = 0$.
\end{itemize}
\end{corollary}

\begin{corollary}
\label{semistrictG2}
If $G$ is a compact simple Lie group with Lie algebra $\g$ and 2-plectic structure $\nu_{k}$,
then there is a semistrict Lie 2-algebra $L(G,k)_\s$
where:
\begin{itemize}
\item the space of 0-chains is $\g^{\ast}$,
\item the space of 1-chains is $\R$,
\item the differential is the exterior derivative $d \maps \R \to
  \g^{\ast}$ (i.e.\ $d=0$), 
\item the bracket is $\Sblank$,
\item the alternator is the identity, hence given by the bilinear
map $S \maps \g^{\ast} \times \g^{\ast} \to \R$ with $S_{\alpha,\beta} = 0$, and
\item the Jacobiator is the trilinear map $J\maps \g^{\ast}\times \g^{\ast}\times 
\g^{\ast}\to \R$ defined by $J_{\alpha,\beta,\gamma} = -\ip{\alpha}\ip{\beta}\ip{\gamma}\nu_{k}$.
\end{itemize}
\end{corollary}
Note that we obtain
a hemistrict Lie 2-algebra $L(G,k)_\h$ and
a semistrict Lie 2-algebra $L(G,k)_\s$ for every nonzero
real number $k$. It is also important to recall that Corollary \ref{isomorphismG}
implies the hemistrict Lie 2-algebra $L(G,k)_\h$ is
isomorphic to the semistrict Lie 2-algebra $L(G,k)_\s$.

Also we see from the proof of Theorem \ref{left_invariant_1-forms}
that there is a nice correspondence between left invariant Hamiltonian
1-forms and left invariant Hamiltonian vector fields which relies on
the isomorphism between $\g$ and its dual space via the inner product
$\innerprod{\cdot}{\cdot}$. As a result, we have the following proposition
which will be useful in the next section:
\begin{prop}
\label{Ham_vect}
If $G$ is a compact simple Lie group with 2-plectic structure
$\nu_{k}$ and $\innerprod{\cdot}{\cdot}$ is the
inner product on the Lie algebra $\g$ of $G$ used in the construction of
$\nu_{k}$, then $\Vect\left(G\right)^L=\g$ and there is an isomorphism of vector spaces
\[\varphi \maps \Vect\left(G\right)^L \stackrel{\sim}{\longrightarrow} \hamL \]
such that $\varphi(v)=k\innerprod{v}{\cdot}$ is the unique left
invariant Hamiltonian 1-form whose Hamiltonian vector field is $v$.
\end{prop}
\begin{proof}
We show only uniqueness since the rest of the proposition follows
immediately from the arguments made in the proof of Theorem
\ref{left_invariant_1-forms}. Let $\alpha \in \g^{\ast}$ and let
$v_{\alpha}$ be its corresponding Hamiltonian vector field.  Suppose
$v_{\alpha}$ also corresponds to any other left invariant Hamiltonian
1-form $\beta \in \g^{\ast}$. Then $d\alpha = -\ip{\alpha}
\nu_{k}=d\beta$ implies $d \left(\alpha -\beta \right)=0$. Hence
$\gamma=\alpha-\beta$ is a closed left invariant 1-form.  Since $\g$
is simple, its first cohomology $H^1\left(\g,\R\right) \cong
\g/\left[\g,\g\right]$ is trivial. The first de Rham cohomology group
of left invariant forms on $G$ is isomorphic to
$H^1\left(\g,\R\right)$. Therefore $\gamma=df$, where $f \in \cinfL$
is a left invariant smooth function. Since $\cinfL$ is the set of
constant functions, it follows that $\gamma=0$ and hence $\alpha$ is
the unique Hamiltonian 1-form corresponding to the Hamiltonian vector
field $v_{\alpha}$.
\end{proof}  
    
\section{The String Lie-2 Algebra}
\label{string_Lie} 

We have described how to construct hemistrict and semistrict Lie
2-algebras from any compact simple Lie group $G$ and any nonzero real
number $k$ using the 2-plectic structure $\nu_{k}$.  Now we show that
these are isomorphic to the `string Lie 2-algebra' of $G$.

It was shown in previous work \cite{HDA6} that semistrict Lie
2-algebras can be classified up to equivalence by data consisting
of:
\begin{itemize}
\item{a Lie algebra $\g$,}
\item{a vector space $V$,}
\item{a representation $\rho \maps \g \to \End\left(V\right)$,}
\item{an element $[j]\in H^{3} \left(\g,V \right)$ of the Lie algebra
  cohomology of $\g$.}   
\end{itemize}
A semistrict Lie 2-algebra $L$ is constructed from this data by
setting the space of 0-chains $L_{0}$ equal to $\g$, the space
1-chains $L_{1}$ equal to $V$, and the differential to be the zero
map: $d=0$. The bracket $\left[\cdot,\cdot \right]
\maps L \otimes L \to L$ is defined to be the Lie bracket on $\g$ in degree
0, and defined in degrees 1 and 2 by:
\[ [x,a]=\rho_{x}(a), \qquad [a,x]=-\rho_{x}(a), \qquad [a,b]=0,\]
for all $x  \in L_{0}$ and $a,b \in L_{1}$. 
The Jacobiator is taken to be any 3-cocycle $j$ representing
the cohomology class $[j]$.

From this classification we can construct the \textbf{string Lie
2-algebra} $\g_{k}^\s$ of a compact simple Lie group $G$ by taking
$\g$ to be the Lie algebra of $G$, $V$ to be $\R$, $\rho$ to
be the trivial representation, and 
\[  j(x,y,z)=k \innerprod{x}{\left[y,z\right]} \]
where $k \in \R$.  When $k \neq 0$,
the 3-cocycle $j$ represents a nontrivial cohomology class.  Note that
since $\rho$ is trivial, the bracket of $\g_{k}^\s$ is trivial in all
degrees except 0.

It is natural to expect that the string Lie 2-algebra is closely
related to the Lie 2-algebra $L(G,k)_\s$ described in Corollary
\ref{semistrictG2}, since both are semistrict Lie 2-algebras built
using solely the trilinear form $\theta_{k}$ on $\g$.  Indeed, this
turns out to be the case:

\begin{theorem}
\label{string_Lie_Thm}
If $G$ is a compact simple Lie group with Lie algebra $\g$ 
and 2-plectic structure $\nu_{k}$, then the string Lie 2-algebra
$\g_{k}^\s$ is isomorphic both to the semistrict Lie 2-algebra $L(G,k)_\s$
and to the hemistrict Lie 2-algebra $L(G,k)_\h$.
\end{theorem}
\begin{proof}
Recall that in degree 0, the bracket of $\g_{k}^\s$ is the Lie bracket $\left
[\cdot,\cdot \right]$ of $\g$ and the bracket of
$L(G,k)_\s$ is the bracket $\Sblank$ introduced in Definition
\ref{semi-bracket.defn}. In all other degrees, both
brackets are the zero map. To simplify notation, 
$\g_{k}^\s$ and $L(G,k)_\s$ will denote both the Lie 2-algebras 
and their underlying chain complexes.
We will show that there exists a Lie 2-algebra
isomorphism between $\g_{k}^\s$ and $L(G,k)_\s$ by constructing
a chain map
\[\phi \maps \g_{k}^\s \to L(G,k)_\s \]
and a chain homotopy
\[\Phi \maps \g_{k}^\s \tensor \g_{k}^\s \to L(G,k)_\s,\]
satisfying the conditions listed in Definition
\ref{homo}. Indeed, we will show that the maps 
$\Sblank \circ \left( \phi \otimes \phi \right)$ and
$\phi \circ \left [\cdot,\cdot \right]$ are actually equal.    

By Proposition \ref{Ham_vect}, there exists a vector space isomorphism
\[\varphi \maps \Vect\left(G\right)^L \stackrel{\sim}{\longrightarrow} \hamL\]
which takes $x \in \Vect\left(G\right)^L$ to the left invariant
Hamiltonian 1-form $\varphi(x)$ whose Hamiltonian vector field is $x$.
The degree 0 components of $\g_{k}^\s$ and $L(G,k)_\s$ are 
$\g$ and $\g^{\ast}$, respectively. From Proposition \ref{Ham_vect} and  
Theorem \ref{left_invariant_1-forms}, we have
\[ \g = \Vect\left(G\right)^{L},  \qquad  \g^{\ast}=\hamL.\]
Using these equalities and the
above isomorphism, we can define $\phi_{0}\maps\g \to \g^{\ast}$ to be
the chain map $\phi$ in degree 0 with
\[\phi_{0}(x)=\varphi(x). \]  
The degree 1 component of both $\g_{k}^\s$ and $L(G,k)_\s$ is
$\R$, and so we define $\phi$ in degree 1 to be the identity map on $\R$.

If $x,y\in \g$, then it follows from Proposition \ref{Ham_vect} that
$\phi_{0}(x)$, $\phi_{0}(y)$, and $\phi_{0} \left([x,y]\right)$ are
the unique left invariant Hamiltonian 1-forms whose
Hamiltonian vector fields are $x$, $y$, and $[x,y]$,
respectively. But Proposition \ref{semi-bracket} implies 
\[d\Sbrac{\phi_0(x)}{\phi_0(y)}= -\iota_{\left[x,y \right]} \nu_{k}.\]
Hence $[x,y]$ is also the Hamiltonian vector field of
$\Sbrac{\phi_{0}(x)}{\phi_{0}(y)}$. It then follows from uniqueness that 
$\Sbrac{\phi_{0}(x)}{\phi_{0}(y)}=\phi_{0} \left([x,y] \right)$.
Therefore the chain maps
\[\Sblank \circ \left( \phi \otimes \phi \right) \maps \g_{k}^\s \otimes
\g_{k}^\s \to L(G,k)_\s \]
and
\[\phi \circ \left [\cdot,\cdot \right] \maps \g_{k}^\s \otimes
\g_{k}^\s \to L(G,k)_\s \]
are equal and hence the chain homotopy $\Phi$ can be taken to be the
identity.  It then follows that the equations in Definition \ref{homo}
hold trivially, and the chain map $\phi$ is invertible by
construction.  Finally, by applying Corollary \ref{isomorphismG}, we
see that $\g_{k}^\s$ is also isomorphic to the hemistrict Lie
2-algebra $L(G,k)_\h$.
\end{proof}

Roytenberg \cite{Roytenberg} has also shown that given a simple Lie
algebra $\g$ and $k \in \R$, one can construct a hemistrict Lie
2-algebra $\g_k^\h$ where:
\begin{itemize}
\item{the space of 0-chains is $\g$,}
\item{the space of 1-chains is $\R$,}
\item{the differential $d$ is the zero map,}
\item{the bracket is the Lie bracket of $\g$ in degree 0 and trivial
  in all other degrees,}
\item{the alternator is the bilinear map 
$S(x,y)= -2k\innerprod{x}{y}$, and}
\item{the Jacobiator is the identity.}
\end{itemize}

He also showed that this hemistrict Lie 2-algebra $\g_{k}^\h$ was
isomorphic to the already known semistrict version of the string Lie
2-algebra, which we are calling $\g_{k}^\s$.  Combining his result
with Corollary \ref{isomorphismG} and Theorem \ref{string_Lie_Thm}, it
becomes clear that we are dealing with the same Lie 2-algebra in four
slightly different guises:
\[\g^{\h}_{k} \cong \g^{\s}_{k} \cong L(G,k)_\h \cong L(G,k)_\s. \]
In particular, we may view the Lie 2-algebras $L(G,k)_\h$ and
$L(G,k)_\s$ as geometric constructions of $\g^{\h}_{k}$ and
$\g^{\s}_{k}$, respectively.

\subsubsection*{Acknowledgements}

We thank Danny Stevenson for suggesting that our construction of Lie
2-algebras from 2-plectic manifolds might yield the string Lie
2-algebra when applied to a compact simple Lie group.  We also thank
Dmitri Roytenberg for helpful conversations.  This work was partially
supported by a grant from The Foundational Questions Institute.

\end{document}